\documentclass[11pt]{article}
\usepackage[utf8]{inputenc}

\usepackage{amssymb}
\usepackage{amsmath}
\DeclareMathOperator*{\argmax}{\arg\!\max}

\usepackage{url}

\usepackage{pgfplots}

\usepackage{graphicx}

\usepackage{nicefrac}

\usepackage{float}

\usepackage{amsthm}
\newtheorem{theorem}{Theorem}
\newtheorem{definition}{Definition}
\newtheorem{corollary}{Corollary}

\newcommand{\trans}[1]{\ensuremath{{#1}^{\scriptscriptstyle \mathsf{T}}}}

\usepackage{tikz}
\usetikzlibrary{arrows}
\tikzset{
 mainNode/.style =
    { circle
    , draw
    }
}

\usepackage{amsmath}
\makeatletter
\renewcommand*\env@matrix[1][*\c@MaxMatrixCols c]{%
\hskip -\arraycolsep
\let\@ifnextchar\new@ifnextchar
\array{#1}}
\makeatother

\usepackage[small]{caption}
\usepackage{subcaption}

\begin{document}

\begin{titlepage}
\begin{center}

{\Huge Overlapping Communities in Complex Networks}
\\[2cm]

\includegraphics[width=0.7\textwidth]{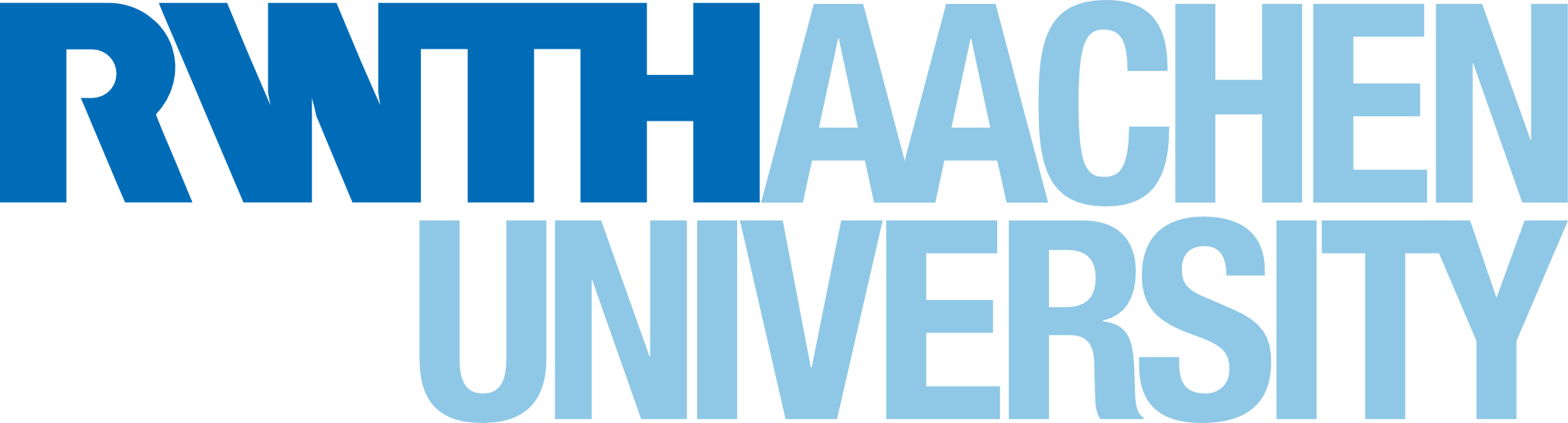}
\\[2cm]

{\huge Jan Dreier}
\\[2cm]

{\Large Bachelor's Thesis}
\\[0.5cm]
{\Large
 Theoretical Computer Science \\
 Department of Computer Science \\
 RWTH Aachen University \\}



\vfill
{\Large 2014}

\end{center}
\end{titlepage}

\newpage
~
\newpage


\null
\vfill
I hereby declare that I have created this work completely on my own and used no other sources or tools than the ones listed, and that I have marked any citations accordingly.

\paragraph{}
Hiermit versichere ich, dass ich die vorliegende Arbeit selbst\"andig verfasst und keine anderen als die angegebenen Quellen und Hilfsmittel benutzt sowie Zitate kenntlich gemacht habe. 

\begin{flushright}
\vspace{12mm}
$\overline{\textit{ Aachen, \today}}$\\
\textit{Jan Dreier}
\end{flushright}
\newpage


\tableofcontents
\newpage


\listoffigures
\newpage


\section{Abstract}
\subsection{English}

Communities are subsets of a network that are densely connected inside and share only few connections to the rest of the network.
The aim of this research is the development and evaluation of an efficient algorithm for detection of overlapping, fuzzy communities.

The algorithm gets as input some members of each community that we aim to discover. We call these members seed nodes.
The algorithm then propagates this information by using random walks that start at non-seed nodes and end as they reach a seed node.
The probability that a random walk starting at a non-seed node $v$ ends at a seed node $s$ is then equated with the probability that $v$ belongs
to the communities of $s$.

The algorithm runs in time
$\tilde{O}(l \cdot m \cdot \log n )$, where $l$ is the number of communities to detect, 
$m$ is the number of edges, $n$ is the number of nodes.
The $\tilde{O}$-notation hides a factor of at most $(\log \log n)^2$.

The LFR benchmark proposed by Lancichinetti et al.\ is used to evaluate the performance of the algorithm.
We found that, given a good set of seed nodes, it is able to reconstruct the communities of a network in a meaningful manner.

\newpage

\subsection{Deutsch}
Eine Community ist eine Untermenge eines Netzwerks, welche intern stark verknüpft ist, jedoch nur wenige Verbingungen zum Rest des Netzwerks besitzt. 
Das Ziel dieser Arbeit ist die Entwicklung und Beurteilung eines effizienten Algorithmus zur Entdeckung von überlappenden, weichen Communitys.

Der Algorithmus erhält als Eingabe Mitglieder jeder Community, die er entdecken soll. Diese Mitglieder nennen wir Seed Knoten.
Der Algorithmus verbreitet diese Information mit Hilfe von Random Walks, welche bei nicht-Seed Knoten anfangen, und enden sobald sie einen Seed Knoten erreichen.
Die Wahrscheinlichkeit, dass ein Random Walk, der bei einem nicht-Seed Knoten $v$ anfängt, bei einem Seed Knoten $s$ endet, 
wird gleichgesetzt mit der Wahrscheinlichkeit, dass $v$ den Communitys von $s$ angehört.

Der Algorithmus läuft in Zeit
$\tilde{O}(l \cdot m \cdot \log n )$, wobei $l$ die Anzahl der Communitys ist, die entdeckt werden sollen;
$m$ und $n$ sind jeweils die Anzahl der Kanten und Knoten.
Die $\tilde{O}$-Notation verbirgt einen Faktor von maximal $(\log \log n)^2$.

Des weiteren haben wir die Performance unseres Algorithmus mit Hilfe des LFR Benchmark von Lancichinetti et al.\ bewertet.
Wir haben herausgefunden, dass unser Algorithmus mit ausreichend Seed Knoten in der Lage ist Communitys sinnvoll zu erkennen.



\newpage
\section{Introduction}

The modern study of networks tries to understand and extract information from complex networks.
It plays a great role in many fields of science, such as biology, sociology and computer science.
One important goal is the detection of community structure.
A network is said to have community structure if its nodes can be separated into sets which are densely connected inside and share only few connections to other sets.
These sets are called communities.
Communities operate mostly independent of the rest of the network and can be analyzed as self-contained entity. 
Also, the interactions between communities describe a meta-network which reveals additional information about the network as a whole.

\textit{Non-overlapping community detection} assigns each node to exactly one community \cite{newman2006} \cite{radicchi2004}.
In contrast to that, \textit{overlapping community detection} allows nodes to belong to multiple communities \cite{gregory2010} \cite{palla2005}.
Community detection may either be \textit{crisp} or \textit{fuzzy}~\cite{gregory2011}.
For crisp detection it is a binary decision whether a node belongs to a community or not.
Fuzzy detection, however, allows nodes to partially belong to (multiple) communities, often indicated by a \textit{belonging factor} ranging between 0 and 1.

In section~\ref{sec:PrevWork} we present current methods for non-overlapping and overlapping detection,
including Newman's modularity, as well as clique percolation and label propagation.
In section~\ref{sec:fundamentals} we discuss the fundamentals needed to understand our model for community detection.
These include absorbing Markov chains (section~\ref{sec:MarkovChains},~\ref{sec:absorbingMarkovChains}), random walks (section~\ref{sec:randomWalks})
and symmetric diagonally dominant linear systems (section~\ref{sec:SDDSystems}).
In section~\ref{sec:communityDetection} we present our model for community detection.
Finally, in section~\ref{sec:experimental} we evaluate it based on the LFR benchmark proposed by Lancichinetti et al.~\cite{lan2009},
which uses random graph models to simulate complex networks.

\subsection{Notation}
In this section we introduce the notation used throughout this thesis.
If $M$ is a matrix then $M_{ij}$ denotes the entry in the $i$th row and $j$th column of $M$. Furthermore, $M_{i*}$ and $M_{*j}$ describe the $i$th row and the $j$th column, respectively.
Given a graph $G=(V,E)$, $V$ is a set of $n$ nodes and $E$ is a set of $m$ edges.
The degree of a vertex $v$ is denoted by $deg(v)$.
The adjacency matrix of $G$ is defined as $A \in \mathbb{R}^{n \times n}$ with
$$
A_{ij} := 
\begin{cases} 
    1        & \mbox{if $v_i$ is adjacent to $v_j$} \\
    0        & \mbox{otherwise}
\end{cases}
$$

\newpage
\section{Previous Work}
\label{sec:PrevWork}

This section serves as a quick overview to some algorithms for overlapping and non-overlapping community detection.
As the field is very diverse, an extensive review is beyond the scope of this thesis. 
For more information, we recommend \cite{xie2011} and \cite{fortunato2010}, two detailed surveys which have been a great source of information for this section.


\subsection{Non-Overlapping Community Detection}

The problem of community detection is not as well defined as it might seem.
It is hard to define a community in a mathematically strict way.
There are many valid definitions and often they differ greatly.
In many cases some sort of arbitrariness and common sense is involved.
One central property most definitions have in common is the assumption that 
there are many edges between members of a community and few edges between members of different communities.

One straightforward definition of a community would be a \textit{clique}, i.e., a fully connected subgraph.
Then the problem of community detection transforms to the well known problem of finding cliques within a graph.
This might be a bit too restrictive though, as it is easy to come up with examples of groups in social networks
which one would consider a community, but in which not everybody knows everybody.

An alternative approach would be to compare the number of connections within the community to the number of connections to the rest of the graph.
A subgraph is said to be a \textit{strong community} if for each vertex the number of edges to vertices within the community exceeds the number of edges to non-community vertices.
In a \textit{weak community} the total number of edges within a community exceeds the total number of outgoing edges.
Radicchi et al.~\cite{radicchi2004} used these concepts in their algorithm for community detection.

Another concept is the detection of communities via a \textit{quality function}. 
A quality function is a function which rates partitions of a graph:
If a partition reveals meaningful community structure it shall be given a higher score.
Community detection then is reduced to finding a partition which maximizes the function.

One of the most important examples of a quality function is the \textit{modularity} proposed by Newman et al.~\cite{newman2006}.
Modularity has been studied extensively and is a widely agreed upon measure for community structure up to today.
It evaluates the goodness of a partition of a network into subgraphs by comparing it to a so called \textit{null model}.

The null model describes a random graph where edges of the original graph are rewired at random, but vertices keep their degree.
It assigns to each pair of vertices a probability that there is an edge between them.
$$
\Pr(\text{edge between vertex $v$ and $w$}) = \frac{k_v \cdot k_w}{2m}
$$
where $k_v$ and $k_w$ are the degree of vertices $v$ and $w$ and $m$ is the total number of edges within the network.

A graph partition has high modularity (i.e., reveals community structure) if the actual number of edges within each subgraph exceeds 
the expected number of edges within each subgraph after edges were rewired according to the null model.
The modularity $Q$ of a graph is defined as
$$
Q = \frac{1}{2m} \sum_{v,w} \left( A_{vw} - \frac{k_v \cdot k_w}{2m} \right) \delta(c_v,c_w)
$$
where $A$ is the adjacency matrix of the network and
$\delta(c_v,c_w)$ is 1 if vertex $v$ and $w$ belong to the same subgraph, otherwise 0.

However, finding a partition which maximizes the modularity is an NP-hard problem~\cite{modularityNPHard}, thus the optimal solution usually is infeasible to find.
Approximate solutions can be obtained by finding eigenvectors in specially crafted matrices~\cite{newmanEigenvec}.
An in depth discussion on modularity and optimization algorithms can be found in~\cite{newmanEigenvec} and \cite{newman2006}.
For more information on non-overlapping community detection we recommend the survey by Fortunato~\cite{fortunato2010}.



\subsection{Overlapping Community Detection}

The traditional approach of non-overlapping detection assigns each vertex to exactly one community;
however, this does not always model the real world. 
In a social network a person may belong to multiple communities (e.g., family, co-workers, sports club).
Overlapping community detection takes this into account by assigning each vertex to one or more communities.
There are many different approaches to overlapping community detection, two of which are presented in this section.

\newpage
\paragraph{Clique Percolation}
The clique percolation method (CPM) proposed by Palla et al.~\cite{palla2005} builds communities based on \textit{$k$-cliques} (cliques of size $k$).
Two $k$-cliques are considered adjacent if they share $k-1$ vertices. 
Communities are identified as unions over all $k$-cliques that can be reached from each other through a series of adjacent $k$-cliques.
They directly correspond to the connected components in a graph of all $k$-cliques, where two $k$-cliques are connected by an edge if they are adjacent.

Since vertices may belong to multiple $k$-cliques, this definition allows for overlapping communities.
Small values for $k$ (between 3 and 6) have been shown to give good results~\cite{palla2005}.
CFinder\footnote{\url{http://cfinder.org}} is an implementation of this algorithm.

\paragraph{Label Propagation}
The underlying idea of label propagation is that vertices adopt the label of its neighboring vertices and form communities based on their label.
Label propagation has been used for both non-overlapping \cite{raghavan2007} \cite{xie2013} and overlapping \cite{gregory2010} community detection.

In the initialization phase of the algorithm each vertex is assigned a unique label.
Then each vertex adopts the label which occurs most frequently within the set of labels of neighboring vertices. Ties are broken at random.
This step is repeated until the vertices have found a consensus on their label.
In the end, all vertices carrying the same label are identified with the same community.
Label updates can happen synchronously (the label of a vertex at update step $i+1$ depends on its neighboring set at step $i$)
or asynchronously (vertices update their label in some fixed order).

The COPRA\footnote{\url{http://www.cs.bris.ac.uk/~steve/networks/software/copra.html}} algorithm by Gregory et al.\cite{gregory2010} extends this
idea to overlapping community detection.
In this algorithm a vertex has a list of labels with corresponding belonging factors between 0 and 1.
In the update step each vertex averages the belonging factors of its neighboring vertices and drops labels whose belonging factor is below some threshold.

\paragraph{}
This section should give a rough idea how diverse community detection algorithms can be.
An extensive overview over a total of 14 current algorithms for overlapping detection, as
well as evaluations and benchmarks can be found in~\cite{xie2011}.

\newpage
\section{Fundamentals}
\label{sec:fundamentals}

In this section we introduce absorbing Markov chains and prove some of their main properties.
Furthermore, we discuss random walks, which are an application of Markov chains and a fundamental building block of our model.
At last, we give a brief overview of near-linear-time solvers for symmetric diagonally dominant linear systems.
The use of such solvers greatly improves the time complexity of our algorithm.

\subsection{Markov Chains}
\label{sec:MarkovChains}

Markov chains are used to model stochastic processes which change over time. 
The initial status of a process is known and the status of the process as time progresses is of interest.

Formally, a Markov chain is a sequence of random Variables $X_1, X_2, X_3, \dots$ which fulfill the Markov property,
namely that the next step only depends on the current step, i.e., 
$\Pr(X_{n+1} = x | X_1 = x_1, X_2 = x_2, \dots, X_n = x_n) = \Pr(X_{n+1} = x | X_n = x_n)$.

The set of possible values for $X_i$ is called the state space of the chain. 
We assume a finite state space $S = \{s_1, \ldots, s_n\}$.
Markov chains are defined by the transition probabilities between these states, which can be represented in graph or matrix form.
An example of both representations of a specific Markov chain can be found in figure~\ref{fig:markovGraph} and~\ref{fig:markovMatrix}.




\paragraph{Probability Distribution}
The state of a Markov chain at a given time is described by a probability distribution $\pi = \trans{(\pi(1), \dots, \pi(n))}$,
where $\pi(i)$ denotes the probability of being at state $s_i$.
Every probability distribution satisfies:
$$
\text{$\pi(i) \ge 0$ for $1 \le i \le n$}
$$
and 
$$
\sum_{i = 1}^n \pi(i) = 1
$$

\paragraph{Stochastic Matrix}
Transitions in a Markov chain lead from one probability distribution to the next and happen in discrete steps. 
Each state has a certain probability to transition to another state.
These probabilities can be expressed by a stochastic matrix $P$, where $P_{ij}$ denotes the probability to reach state $s_j$ from state $s_i$.
A transition to the next probability distribution can now be expressed by multiplying with $P$ from the right: $\pi_{i+1} = \pi_i P$.
$P$~is a square matrix with
$$
\text{$P_{ij} \ge 0$ for $1 \le i,j \le n$}
$$
and 
$$
\text{$\sum_{j = 1}^n P_{ij} = 1$ for $1 \le i \le n$}.
$$


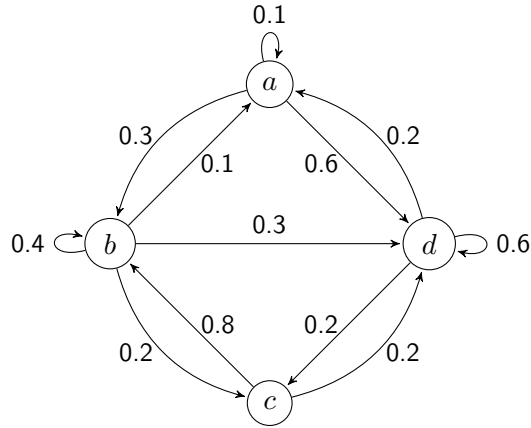
\begin{figure}
\centering
\begin{tikzpicture}
    [ auto
    , ->
    , >=stealth'
    , shorten >=1pt
    , node distance=3cm
    ]

    \node[mainNode] (1) {$a$};
    \node[mainNode] (2) [below left of=1] {$b$};
    \node[mainNode] (3) [below right of=2] {$c$};
    \node[mainNode] (4) [below right of=1] {$d$};

    \path[every node/.style={font=\sffamily\small}]
    (1) edge node [left] {0.6} (4)
        edge [bend right] node[left] {0.3} (2)
        edge [loop above] node {0.1} (1)
    (2) edge node [right] {0.1} (1)
        edge node {0.3} (4)
        edge [loop left] node {0.4} (2)
        edge [bend right] node[left] {0.2} (3)
    (3) edge node [right] {0.8} (2)
        edge [bend right] node[right] {0.2} (4)
    (4) edge node [left] {0.2} (3)
        edge [loop right] node {0.6} (4)
        edge [bend right] node[right] {0.2} (1)
    ;
\end{tikzpicture}
\caption{Graph representation of a Markov chain}
\label{fig:markovGraph}
\end{figure}

\begin{figure}
\centering
\[
P =
    \bordermatrix{
          &   a   & b   & c   & d   \cr 
        a &   0.1 & 0.3 & 0   & 0.6 \cr 
        b &   0.1 & 0.4 & 0.2 & 0.3 \cr 
        c &   0   & 0.8 & 0   & 0.2 \cr 
        d &   0.2 & 0   & 0.2 & 0.6
    }
\]
\caption{Transition matrix of the same Markov chain}
\label{fig:markovMatrix}
\end{figure}

\paragraph{}
A Markov process needs an initial distribution $\pi_0$.
Since $\pi_{i+1} = \pi_i P$, the probability distribution after multiple steps can be expressed with powers of $P$:
$$\pi_k = \pi_{k-1}P = \pi_{k-2} P^2 = \dots = \pi_0 P^k$$
Furthermore, the probability to reach state $s_j$ from state $s_i$ in $k$ steps equals the $ij$th entry in $P^k$.

\subsection{Absorbing Markov Chains}
\label{sec:absorbingMarkovChains}

Absorbing Markov chains are special Markov chains with each state either being absorbing or transient.
The probability of leaving an absorbing state is zero.
Transient states have a nonzero probability to reach at least one absorbing state after a finite number of steps.

\paragraph{}
The stochastic matrix of an absorbing Markov chain with $\sigma$ absorbing states $\{s_1, \ldots, s_\sigma\}$ 
and $\tau$ transient states $\{t_1, \ldots, t_\tau\}$ can be written as:

$$
P=
\begin{pmatrix}[c|c]
Q & R\\ \hline
0 & I
\end{pmatrix}
$$

Where
\begin{itemize}
\item $Q \in \mathbb{R}^{\tau \times \tau}$, $Q_{ij}$ denotes the probability to go from transient state $t_i$ to transient state $t_j$. 
\item $R \in \mathbb{R}^{\tau \times \sigma}$, $R_{ij}$ denotes the probability to go from transient state $t_i$ to absorbing state $s_j$.
\item $0 \in \mathbb{R}^{\sigma \times \tau}$ is the zero matrix.
\item $I \in \mathbb{R}^{\sigma \times \sigma}$ is the identity matrix.
\end{itemize}

Using simple matrix transformations, it can be shown that:
\begin{equation}
\label{grid}
P^k=
\begin{pmatrix}[c|c]
    Q^k & \sum\limits_{i=0}^{k-1} Q^{i} R \\ \hline
    0   & I
  \end{pmatrix}
\end{equation}

An important property of an absorbing Markov chain is that it always converges to a steady probability distribution.
For a general Markov chain this is not the case.
The final result of the community detection algorithm presented in section~\ref{sec:communityDetection}
will be a steady probability distribution of an absorbing Markov chain,
so the following theorem is crucial.

\begin{definition}
$P^\infty := \lim \limits_{k \rightarrow \infty}{P^k}$ is the transition matrix of an absorbing Markov chain for an infinite number of steps.
\end{definition}

\begin{theorem} 
\label{mainabsorb}
$P^\infty$ is well defined and
$$
    P^\infty=
    \begin{pmatrix}[c|c]
        0 & (I-Q)^{-1}R \\ \hline
        0 & I
      \end{pmatrix}
$$
\end{theorem}

\begin{proof}
According to equation \eqref{grid}, we need to show
\begin{equation}
    \label{qpowerzero}
    \lim \limits_{k \rightarrow \infty}{Q^k} = 0
\end{equation}
and
\begin{equation}
    \label{qsumb}
    \sum\limits_{k=0}^{\infty} Q^{k} R = (I-Q)^{-1} R.
\end{equation}
\eqref{qsumb} is a direct consequence of \eqref{qpowerzero} because the geometric series converges if the geometric sequence converges.
Since $Q$ is non-negative, a sufficient criterion for convergence of the sequence \eqref{qpowerzero} is the existence of an $l$ so that
for all  $1 \le i \le \tau$:
$$
    \sum_{j = 1}^\tau [Q^l]_{ij} < 1
$$
Because $P$ is an absorbing Markov chain, we can find an $l$ so that 
for each transient state $t_i$ there is a non-zero probability $p_i$ to reach an absorbing state after $l$ steps.
$P^l$ is stochastic, so for all $1 \le i \le \tau$:
$$
    \sum_{j = 1}^\tau [Q^l]_{ij} = 1 - p_i < 1
$$

\end{proof}

$P^\infty$ is well defined and describes an infinite number of steps of an absorbing Markov process.
Because of \eqref{qpowerzero}, the transition probabilities to transient states in $P^\infty$ are all zero,
consequently each state will be absorbed with a probability of one.

\paragraph{}
This leads us to the final result of this section:
\begin{corollary}
\label{corolabsorb}
The probability that a transient state $t_i$ is absorbed in an absorbing state $s_j$ equals the $ij$th entry in $(I-Q)^{-1}R$.
\end{corollary}
More information on absorbing Markov chains can be found in the chapter~11 of
Grinstead and Snell's book \textit{Introduction to Probability}~\cite{introProbability}.

\subsection{Random Walks}
\label{sec:randomWalks}

\textit{Random walks} are an application of Markov chains and our model presented in section~\ref{sec:communityDetection} heavily relies on them.
A random walk $(v_0, v_1, v_2, \ldots)$ in a graph $G = (V,E)$ is a path generated by a stochastic process.
$v_0$ is the start vertex and $v_{i+1}$ is chosen uniformly random among the adjacent vertices of $v_i$.

\paragraph{}
The random walk can be described by a Markov chain:
The state space corresponds to the vertices in $G$ and
the transition matrix $P$ is defined as
$$
P_{ij} =\begin{cases}
    \frac{1}{deg(v_i)}, & \text{if $(v_i,v_j) \in E$} \\
    0, & \text{otherwise}.
  \end{cases}
$$
The probability to be at vertex $j$ after a random walk of $k$ steps starting at vertex $i$ equals the $ij$th entry in $P^k$.

\subsection{Symmetric Diagonally Dominant Linear Systems}
\label{sec:SDDSystems}

A linear system is a set of linear equations
$$A x = b$$
where $A$ is a matrix and $b$, $x$ are vectors of appropriate size.
$A$ and $b$ are known and a value for $x$ which satisfies all equations is of interest.
Countless problems from numerical mathematics, engineering and science can be reduced to solving linear systems.

Fast algorithms to solve such systems are of particular interest.
A naive approach using Gaussian elimination on an $n \times n$ matrix takes running time $O(n^3)$.
Given that in real world applications, matrices easily contain millions of entries, such a running time is impractical.
However, if the matrix contains a certain structure this can be exploited and far better running times can be achieved.


\paragraph{}
An important matrix for graph theory is a graph's Laplacian.
\begin{definition}
Given a graph $G$ with $n$ vertices, its Laplacian is defined as $L \in \mathbb{R}^{n \times n}$ with
$$
L_{ij} := 
\begin{cases} 
    deg(v_i)  & \mbox{if $i = j$} \\ 
    -1        & \mbox{if $i \ne j$ and $v_i$ is adjacent to $v_j$} \\
    0         & \mbox{otherwise}
\end{cases}
$$
\end{definition}
Laplacians play a great role in many fields, including computer graphics or scientific computing.
A graph's Laplacian reveals many interesting properties of a graph.
For example, the multiplicity of the eigenvalue zero gives the number of connected components in a graph.
The algorithm we present in section~\ref{sec:algorithm} will rely on solving linear systems which are almost Laplacian.
Many applications of Laplacians can be found in Fan Chung's book \textit{Spectral Graph Theory}~\cite{chungfan}.


\begin{definition}
An $n \times n$ matrix $A$ is said to be symmetric diagonally dominant (SDD) if it is symmetric and
for all $1 \le i \le n$:
$$
|A_{ii}| \ge \sum_{1 \le j \le n ,\ i \ne j} |A_{ij}|
$$
\end{definition}

A graph's Laplacian is an SDD matrix.
Spielman and Teng recently had a major breakthrough when they showed that 
SDD linear systems can be solved in near-linear time~\cite{ST04,EEST05,ST08}.
Spielman and Teng's algorithm (the ST-solver) combines numerical mathematics and graph theory 
to iteratively produce a sequence of approximate solutions which converge to the exact solution.
The performance of such an iterative system is measured in terms of the time taken to reduce 
an appropriately defined approximation error by a constant factor. The time 
complexity of the ST-solver was reported to be at least $O(m \log^{15} n)$, where $m$ is the number of nonzero entries~\cite{KMP11}.  
Koutis, Miller and Peng~\cite{KMP10,KMP11} developed a simpler and faster algorithm 
for finding $\varepsilon$-approximate solutions to SDD systems in time 
$\tilde{O}(m \log n \log (1/\varepsilon) )$, where the $\tilde{O}$-notation hides 
a factor that is at most $(\log \log n)^2$. A highly readable account 
on SDD systems is the monograph by Vishnoi~\cite{Vis13}. We summarize the 
main result, which we use as a black-box.  

\begin{theorem} 
\label{SDD_systems} 
{\textrm{\cite{KMP11,Vis13}}}
Given a system of linear equations $Ax=b$, where $A$
is an SDD matrix, there exists an algorithm to compute $\tilde{x}$  
such that:
    \[
        \|\tilde{x} - x\|_A \leq \varepsilon \|x\|_A, 
    \]
where $\|y\|_A := \sqrt{\trans{y} A y}$. The algorithm runs in 
time $\tilde{O}(m \cdot \log n \cdot \log (1 / \varepsilon) )$ time, where $m$ is the number of non-zero 
entries in $A$. The $\tilde{O}$-notation hides a factor of at most $(\log \log n)^2$.
\end{theorem} 


This ability to solve SDD linear systems (exp. Laplacians) fast has been used to
obtain many nearly-linear-time algorithms for applications such as semi-supervised learning, image processing or web-spam detection~\cite{lapparad}.
Or as Erica Klarreich puts it in her article \textit{Network Solutions}\footnote{\url{http://www.simonsfoundation.org/mathematics-and-physical-science/network-solutions}}:
``A new breed of ultrafast computer algorithms offers computer scientists a novel tool to probe the structure of large networks."



\newpage
\section{Our Model for Community Detection}
\label{sec:communityDetection}

Now that we introduced all necessary fundamentals we can present our model for community detection.
Our goal is to find multiple, possibly overlapping communities within a network.

The algorithm we present needs to know some of the members of each community we want to discover.
These members are called \textit{seed nodes} and may belong to multiple communities.
They need to be selected by the user and are handed to the algorithm as part of its input.
The algorithm then uses random walks to extend the partial community-information given by the seed nodes to the rest of the network.

One specific application would be the detection of the political leanings in a social network like Facebook.
The underlying assumption is that people are more likely to interact with people who have the same political beliefs. 
If one knows for some members of the network if they are either conservative or liberal but has no information about the rest of the network
one may use the algorithm to decide for each member of the network if he or she is rather conservative or liberal.

The need for extra information via seed nodes sets this method apart from most other methods.
This can be a disadvantage, as for many networks this information may not be present.
But if a proper set of seed nodes can be specified our algorithm is quite flexible.
By choosing a certain set of seed nodes the user can guide the algorithm to extracting very specific communities.

We say a node has \textit{high affinity} to a community if it belongs to it and \textit{low affinity} if it does not.
Intermediate affinity values are possible and correspond to a partial belonging.
The affinity of all seed nodes needs to be known to the algorithm beforehand.
For all other nodes, the \textit{non-seed nodes}, we want deduce the affinity to each community.
We will use information given by the seed node's affinity and the network structure.

The fundamental idea is that non-seed nodes should adopt the affinities of seed nodes within their close proximity.
But how do we define proximity? 
A naive approach would be to pick the seed node with the shortest path distance and adopt its affinities,
but this does not model community structure very well:
Consider a social network where people correspond to nodes and edges correspond to social interaction.
Two people who do not know each other may have one friend in common or they may have several friends in common.
In both cases the shortest path distance would be 2, but several common friends would be a much stronger indicator that these two people belong to the same community.

This is why we use a different approach: We define a proximity measure based on random walks.
The random walk starts at a non-seed node, traverses through the graph, and ends as soon as it reaches a seed node.
The non-seed node then adopts the affinities of the seed nodes the random walk is likely to reach first.

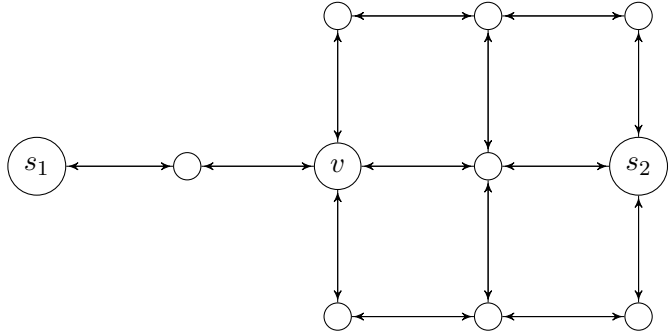
\begin{figure}
\centering
\begin{tikzpicture}
    [ auto
    , ->
    , >=stealth'
    , shorten >=1pt
    , node distance=2cm
    ]
    \node[mainNode] (1) {$s_1$};
    \node[mainNode] (2) [right of=1] {};
    \node[mainNode] (3) [right of=2] {$v$};
    \node[mainNode] (4) [right of=3] {};
    \node[mainNode] (5) [right of=4] {$s_2$};
    \node[mainNode] (3a) [above of =3] {};
    \node[mainNode] (4a) [above of =4] {};
    \node[mainNode] (5a) [above of =5] {};
    \node[mainNode] (3b) [below of =3] {};
    \node[mainNode] (4b) [below of =4] {};
    \node[mainNode] (5b) [below of =5] {};

    \path[every node/.style={font=\sffamily\small}]
    (1) edge node {} (2)
    (2) edge node {} (3)
    (3) edge node {} (4)
    (4) edge node {} (5)

    (5) edge node {} (4)
    (4) edge node {} (3)
    (3) edge node {} (2)
    (2) edge node {} (1)

    (3) edge node {} (3a)
    (4) edge node {} (4a)
    (5) edge node {} (5a)
    (3) edge node {} (3b)
    (4) edge node {} (4b)
    (5) edge node {} (5b)

    (3a) edge node {} (3)
    (4a) edge node {} (4)
    (5a) edge node {} (5)
    (3b) edge node {} (3)
    (4b) edge node {} (4)
    (5b) edge node {} (5)

    (3a) edge node {} (4a)
    (4a) edge node {} (5a)
    (5a) edge node {} (4a)
    (4a) edge node {} (3a)
    (3b) edge node {} (4b)
    (4b) edge node {} (5b)
    (5b) edge node {} (4b)
    (4b) edge node {} (3b)
    ;
\end{tikzpicture}
\caption[Comparison of random walk distance and shortest path distance in a graph.]
    {Example:
    Let $s_1 \in C_1$ and $s_2 \in C_2$ be seed nodes which belong to two implicit communities $C_1$ and $C_2$. All other nodes are non-seed nodes.
    $(s_1, v)$ and $(s_2, v)$ both have a shortest path distance of 2.
    $v$ apparently is more likely to be in the same community as $s_2$ than $s_1$.
    It can be shown that a random walk starting at $v$ reaches $s_1$ with a probability of $\nicefrac{1}{3}$ and $s_2$ with a probability of $\nicefrac{2}{3}$.
    As a result, $v$ has an affinity of $\nicefrac{1}{3}$ to $C_1$ and $\nicefrac{2}{3}$ to $C_2$.
}
\label{fig:randWalkExample}
\end{figure}

Let us assume that a network, indeed, has community structure and contains some hidden communities $C_1, \dots, C_l$ which we want to discover.
For all seed nodes we know which of the communities they belong to. 
A defining property of a community is that it has many inner edges and few leaving edges,
so a random walk starting a node within a community $C_i$ should have a relatively high probability of staying within this community.
This means that the probability that a random walk starting at a non-seed node reaches a seed node in $C_i$ should be higher 
if the non-seed node lies within $C_i$ than if it lies outside.
After all, a non-seed node should adopt the affinities of seed nodes within the same community and thus be assigned correctly.
As a result, the communities we detect should be conform with the community structure of the graph and the seed nodes.
Figure~\ref{fig:randWalkExample} shows the advantage of random walks over the shortest path distance in a small example.

\subsection{Single Community Detection}
\label{sec:singleCommunityDetection}

In this subsection we formally define our model for finding a single community.
Later in section~\ref{sec:MultipleCommunityDetection}, we extend this model to multiple communities.

\paragraph{}
The input to our problem is an undirected, connected graph $G=(V,E)$ and a set of seed nodes $\emptyset \neq S \subset V$.
We want to detect the community $C$.
For all seed nodes $s \in S$ the affinity to $C$ is part of the input.
It is called $\beta(s)$ and may range between 0 and 1. 
$\beta(s) = 0$ means, that $s$ does not belong to $C$ and $\beta(s) = 1$ means that $s$ belongs to $C$.
Intermediate values are possible and correspond to a partial belonging.
The algorithm returns the community by assigning an affinity $\beta(v)$ to all $v \in V \setminus S$.

\paragraph{}
Since random walks should end as soon as they reach a seed node, we transform $G$ into a new graph $G'$ as follows:
First, we make the graph directed by replacing each undirected edge with two directed edges.
Then for each seed node, we remove its outgoing edges and add a self-loop.
The procedure is illustrated in figure~\ref{fig:modgraph}.
From now on, we only work with $G'$.


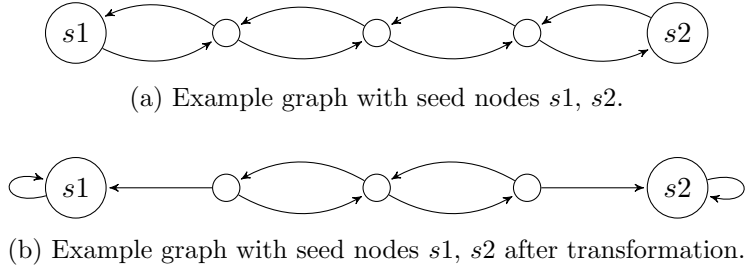
\begin{figure}
\centering
\begin{subfigure}{\textwidth}
    \centering
    \begin{tikzpicture}
        [ auto
        , ->
        , >=stealth'
        , shorten >=1pt
        , node distance=2cm
        ]
        \node[mainNode] (1) {$s1$};
        \node[mainNode] (2) [right of=1] {};
        \node[mainNode] (3) [right of=2] {};
        \node[mainNode] (4) [right of=3] {};
        \node[mainNode] (5) [right of=4] {$s2$};

        \path[every node/.style={font=\sffamily\small}]
        (1) edge [bend right] node {} (2)
        (2) edge [bend right] node {} (3)
        (3) edge [bend right] node {} (4)
        (4) edge [bend right] node {} (5)

        (5) edge [bend right] node {} (4)
        (4) edge [bend right] node {} (3)
        (3) edge [bend right] node {} (2)
        (2) edge [bend right] node {} (1)
        ;
    \end{tikzpicture}
    \caption{Example graph with seed nodes $s1$, $s2$.}
\end{subfigure}
\\[0.5cm]
\begin{subfigure}{\textwidth}
    \centering
    \begin{tikzpicture}
        [ 
          auto
        , ->
        , >=stealth'
        , shorten >=1pt
        , node distance=2cm
        ]
        \node[mainNode] (1) {$s1$};
        \node[mainNode] (2) [right of=1] {};
        \node[mainNode] (3) [right of=2] {};
        \node[mainNode] (4) [right of=3] {};
        \node[mainNode] (5) [right of=4] {$s2$};

        \path[every node/.style={font=\sffamily\small}]
        (1) edge [loop left]  node {} (1)
        (5) edge [loop right] node {} (5)

        (2) edge [bend right] node {} (3)
        (3) edge [bend right] node {} (4)
        (4) edge node {} (5)

        (4) edge [bend right] node {} (3)
        (3) edge [bend right] node {} (2)
        (2) edge node {} (1)
        ;
    \end{tikzpicture}
    \caption{Example graph with seed nodes $s1$, $s2$ after transformation.}
\end{subfigure}
\caption[Removal of outgoing edges of seed nodes in a graph]
{Remove outgoing edges and add self-loop for all seed nodes in an example graph. A random walk reaching $s1$ or $s2$ will stay there forever.}
\label{fig:modgraph}.
\end{figure}

\begin{theorem}
The random walks in $G'$ define an absorbing Markov chain.
\end{theorem}

\begin{proof}
Follows directly from the definition of absorbing Markov chains in section~\ref{sec:absorbingMarkovChains}.
There are no outgoing edges for seed nodes, hence, they represent absorbing states.
For a non-seed node $v$ let $s \in S$ be the seed node nearest to it.
Since $G$ is connected, there is one such node in $S$.
In the undirected graph $G'$, there still is finite a path from $v$ to $s$.
After all, random walks in $G'$ define an absorbing Markov chain with absorbing states $S$ and transient states $V \setminus S$.
\end{proof}

\paragraph{}
We want a non-seed node $v$ to adopt the affinity of a seed node $s$ if the probability that a random walk starting at $v$ is absorbed at $s$ is high.
Let $E_{v \to s}$ be the event that an infinite random walk in $G'$ starting at $v$ is absorbed in $s$.
According to Theorem~\ref{mainabsorb}, every infinite random walk will be absorbed and
$\Pr(E_{v \to s})$ corresponds to an entry in $P^\infty$, where $P$ is the transition matrix for random walks in $G'$, 
so $\Pr(E_{v \to s})$ is well defined and ${\sum_{s \in S} \Pr(E_{v \to s}) = 1}$.

\begin{definition}
We define for all non-seed nodes $v \in V \setminus S$:
$$ \beta(v) := \sum_{s \in S} \beta(s) \Pr(E_{v \to s}) $$
\end{definition}

\paragraph{}
Thus $\beta(v)$ is a convex combination of $\{ \beta(s) | s \in S \}$ with weights assigned according to 
the probability that a random walk starting at $v$ is absorbed at a certain seed node.

\subsection{Algorithm}
\label{sec:algorithm}

In this section we present an efficient algorithm which calculates $\beta(v)$ for all $v \in V \setminus S$.

\begin{theorem} 
Given a graph $G = (V,E)$ and seed nodes $S$, each with an affinity $\beta$,
there exists an algorithm which computes $\beta(v)$ for all $v \in V \setminus S$
with running time 
$\tilde{O}(m \cdot \log n )$, where $m := |E|$ and $n := |V|$.
The $\tilde{O}$-notation hides a factor of at most $(\log \log n)^2$.
\end{theorem} 

\begin{proof}
The input to the algorithm is a graph $G = (V,E)$, the seed nodes $S$ along with their affinity $\beta$.
Let $n := |V|$, $\sigma := |S|$.
The seed nodes are labeled $S = \{s_1, \dots s_\sigma\} $
and the non-seed nodes are labeled $V \setminus S = \{v_1, \dots, v_{n-\sigma}\}$
The stochastic matrix $P \in \mathbb{R}^{n \times n}$ can be written as:
\[
P=
\begin{pmatrix}[c|c]
Q & R\\ \hline
0 & I
\end{pmatrix}
\]
with
$Q \in \mathbb{R}^{n - \sigma \times n - \sigma}$,
$R \in \mathbb{R}^{n - \sigma \times \sigma}$.

\paragraph{}
We want to calculate $\beta(v_i)$ for a non-seed node $v_i$.
According to Corollary~\ref{corolabsorb}, $\Pr(E_{v_i \to s_j}) = B_{ij}$ with $B := (I-Q)^{-1}R$. This yields:
$$ \beta(v_i) = \sum_{j = 1}^{\sigma} \beta(s_j) B_{ij} $$
The values $\beta(v_1), \dots, \beta(v_{n-\sigma})$ can be expressed as a vector.
The matrix multiplication is linear, so it can be pulled out of the sum.
\[
\begin{split}
\begin{pmatrix}
\beta(v_1) \\
\vdots     \\
\beta(v_{n-\sigma})
\end{pmatrix}
& = \sum_{j = 1}^{\sigma} \beta(s_j) B_{*j}            \\
& = \sum_{j = 1}^{\sigma} \beta(s_j) (I-Q)^{-1}R_{*j}  \\
& = (I-Q)^{-1} \sum_{j = 1}^{\sigma} \beta(s_j)R_{*j}
\end{split}
\]
After all, $\beta(v_1), \dots, \beta(v_{n-\sigma})$ can be obtained by solving the linear system
\begin{equation}
\label{linsys}
(I-Q)x = b
\end{equation}
with 
$$
x := 
\begin{pmatrix}
    \beta(v_1) \\
    \vdots     \\
    \beta(v_{n-\sigma})
\end{pmatrix}
$$
and
$$
b := \sum_{j = 1}^{\sigma} \beta(s_j)R_{*j}
$$

Now, $(I-Q)$ has full rank, is diagonally dominant, but not symmetric.
As a next step we will modify both sides of the equation to get an almost Laplacian SDD system with the same solution, which can be solved efficiently.

\paragraph{}
$P$ describes a random walk and is defined as
$$
P_{ij} =\begin{cases}
    \frac{1}{deg(v_i)}, & \text{if $(v_i,v_j) \in E$} \\
    0, & \text{otherwise}.
  \end{cases}
$$
$Q$ is the $(n-\sigma) \times (n-\sigma)$-submatrix of $P$ which holds the transition probabilities from transient to transient vertices.
We define $A$ as the adjacency matrix of $G[v_1, \dots, v_{n-\sigma}]$, the subgraph induced by all transient vertices, and
$D$ as the diagonal matrix with $D_{ii}$ being the degree of $v_i$ in $G$.
Then $Q = D^{-1}A$ and we can rewrite \eqref{linsys} as
$$
    (I-D^{-1}A)x = b
$$
If we multiply with $D$ from the left we get the equivalent system
$$
    (D-A)x = Db.
$$
$D$ and $A$ both are symmetric and $\sum_{j=1}^{n-\sigma}A_{ij}$ equals 
the degree of vertex $v_i$ in $G[v_1, \dots, v_{n-\sigma}]$ whereas $D_{ii}$ holds the degree of $v_i$ in $G$,
so $(D-A)$ is an SDD matrix.

\paragraph{}
$Db$ and $(D-A)$ can be constructed in liner time in the number of non-zero entries.
As a result, the running time of the algorithm is dominated by the near-linear running time of the SDD-solver described in Theorem~\ref{SDD_systems}.
\end{proof}

\subsection{Multiple Overlapping Community Detection}
\label{sec:MultipleCommunityDetection}

Our model for finding a single community can be naturally extended to multiple overlapping communities.
The input is an undirected, connected graph $G=(V,E)$ and a nonempty set of seed nodes $S \subset V$, each with an affinity $\beta$.
We want to detect communities $C_1, \dots, C_l$, so $\beta$ is an $l$-tuple $\beta = (\beta_1, \dots, \beta_l)$
where $0 \le \beta_i(s) \le 1$ describes the affinity of seed node $s$ to community $C_i$.


\begin{definition}
    The graph $G'$, the random walk, and $E_{v \to s}$ are defined as before.
    We define for all non-seed nodes $v \in V \setminus S$:
    $$
        \beta(v) := (\beta_1(v), \dots, \beta_l(v))
    $$
    where
    $$ 
        \beta_i(v) := \sum_{s \in S} \beta_i(s) \Pr(E_{v \to s})
    $$
\end{definition}

\begin{theorem} 
    Given a graph $G = (V,E)$ with $l$ communities, seed nodes $S$ along with their affinity $\beta=(\beta_1, \dots, \beta_l)$,
    there exists an algorithm which computes $\beta(v)$ for all $v \in V \setminus S$
    with running time 
    $\tilde{O}(l \cdot m \cdot \log n) )$, where $m := |E|$ and $n := |V|$.
    The $\tilde{O}$-notation hides a factor of at most $(\log \log n)^2$.
\end{theorem}

\begin{proof}
    The values $\beta_i(v_1), \dots, \beta_i(v_{n-\sigma})$, $1 \le i \le l$ can be calculated by the algorithm for single community detection.
    So the running time for finding $l$ communities is $l$ times the time needed to find a single one.
\end{proof}

\paragraph{}
An interesting property is, that if the affinity of all seed nodes sum up to a certain value the affinities of non-seed nodes do as well.
This way, one can use a model where each vertex has a summed affinity of 1, which is individually distributed among different communities.

\begin{corollary} 
    If $\sum_{i=1}^{l} \beta_i(s) = c$ for all $s \in S$, then $\sum_{i=1}^{l} \beta_i(v) = c$ for all $v \in V$.
\end{corollary} 

\begin{proof}
    \[
    \begin{split}
          \sum_{i=1}^{l} \beta_i(v)
        & = \sum_{i=1}^{l} \sum_{s \in S} \beta_i(s) \Pr(E_{v \to s}) \\
        & = \sum_{s \in S} \sum_{i=1}^{l} \beta_i(s) \Pr(E_{v \to s}) \\
        & = \sum_{s \in S} c \Pr(E_{v \to s}) = c
    \end{split}
    \]
\end{proof}

\newpage
\section{Experimental Results}
\label{sec:experimental}

In this section we use the LFR benchmark proposed by Lancichinetti et al.~\cite{lan2009} to evaluate the quality of our model.
Even though our model is capable of overlapping community detection, to keep things simple, we focus on non-overlapping community detection only.

\subsection{Implementation}

Our reference implementation is written in the C++ programming language.
It uses the Boost Graph Library\footnote{\url{http://www.boost.org/doc/libs/1_55_0/libs/graph}} for graph manipulation and traversal,
and the Eigen library\footnote{\url{http://eigen.tuxfamily.org/}} to solve numerical problems.
Since it is a very recent finding that SDD systems can be solved in near-linear time, there is no stable implementation of this fast algorithm yet.
Instead, we use a direct sparse Cholesky decomposition\footnote{\url{http://eigen.tuxfamily.org/dox-devel/group__SparseCholesky__Module.html}}.
This algorithm does not have a near-linear time complexity, but has been optimized for real world applications and performs quite well in practice.
In our example the detection of 200 communities in a graph consisting of 10000 vertices and 150000 edges on a consumer laptop (Intel Core i3, 4GB RAM) takes approximately 30 seconds.

\subsection{Benchmark}

The LFR benchmark is used widely for evaluation of community detection algorithms~\cite{xie2011}~\cite{lan2009}.
Node degrees and community sizes of many real world networks follow a \textit{power law distribution}~\cite{clauset2009powerlaw} 
and the LFR benchmark aims to simulate these kinds of networks.
A value $x$ is said to obey a power law distribution if it occurs with a probability $p(x) \propto x^{-\alpha}$ for a constant parameter $\alpha$.
Usually, $\alpha$ ranges between 2 and 3. This means low values for $x$ are common and high values are rare.

\newpage

The LFR benchmark generates a random graph and assigns each vertex to exactly one community. This process can be controlled by the following parameters:
\begin{itemize}
\item $N$ controls the number of nodes in a graph.
\item $\langle k \rangle$ controls the average node degree.
\item $\gamma$ and $\beta$ control the exponent of the power law distribution of node degree and the community size, respectively.
\item $\mu$ is called the \textit{mixing parameter} and controls how interweaved the communities are. 
      Each node shares a fraction $1-\mu$ of its edges with other nodes of the same community and a fraction of $\mu$ with nodes that belong to a different community.
\end{itemize}
An illustration of a graph generated by the LFR benchmark is shown in figure~\ref{fig:lfrGraph}.

\paragraph{Setup}
As the first step in our test setup, the LFR benchmark generates a graph $G=(V,E)$ consisting of communities $C_1, \dots, C_l$ and assigns each vertex to one community $C_i$.
Then a random subset $S \subset V$ is chosen as a set of seed nodes.
For each seed node $s \in S$ $\beta(s) = (\beta_1(s), \dots, \beta_l(s))$ is chosen with:
$$
\beta_i(s) = 
\begin{cases} 
    1 &\mbox{if vertex $s$ belongs to community $C_i$} \\ 
    0 &\mbox{otherwise}
\end{cases}
$$
The fraction of seed nodes is controlled by parameter $\sigma$ so that $|S| = \sigma |V|$.
The number of communities in the LFR-graph appears to be linear in the size of graph, so for a fixed $\sigma$ the number of seed nodes remains constant for different graph sizes.

Next, community detection is performed.
The multiple community detection algorithm from section~\ref{sec:MultipleCommunityDetection} returns for each non-seed node $v \in V \setminus S$
and affinity vector $\beta(v) = (\beta_1(v), \dots, \beta_l(v))$. We assign vertex $v$ to the community with the greatest affinity value, i.e.,
$v$ is assigned to $C_i$ where $i = \argmax_{1 \le i \le l} \beta_i(v)$.
We call the original community of vertex $v$ chosen by the LFR benchmark $c(v)$ and the community returned by the algorithm $\hat{c}(v)$.
As an intuitive quality measure $Q$, we use the fraction of correctly assigned vertices.
$Q$ ranges between 0 (bad) and 1 (good).
$$
Q = \frac{|\{v \in V | c(v) = \hat{c}(v) \}|}{|V|}
$$
We calculate this quality measure for our community detection algorithm on various graphs.
Different values for the parameters $\gamma$, $\beta$, $\langle k \rangle$, $\sigma$ and $\mu$ are chosen to simulate a wide range of scenarios.
For each setting 100 runs were performed and the average result for $Q$ was taken. 
The result is shown in figure~\ref{fig:resultBenchmark}.
All benchmark-parameters are again summarized in figure~\ref{fig:summaryBenchmark}.


\newcommand{\varianceplot}[6]{
    \begin{tikzpicture}[scale = 0.75]
        \begin{axis}[
            width=1.3\textwidth,
            title={$N = 500$, $\gamma = #1$, $\beta = #2$ and $\sigma$ = #3, $\langle k \rangle$ = #4, $\mu$ = #5},
            xlabel={quality $Q$},
            ylabel={relative frequency},
            ymin=0, ymax=0.2,
            ytick={0,0.1,0.2},
            legend pos=south west,
            ymajorgrids=true,
            grid style=dashed,
        ]

        \addplot+[ybar] plot coordinates{#6};
     
        \end{axis}
    \end{tikzpicture}
}

\newcommand{\awesomeplot}[6]{
    \begin{tikzpicture}[scale = 0.75]
        \begin{axis}[
            width=1.3\textwidth,
            title={$N = 500$, $\gamma = #1$, $\beta = #2$ and $\sigma$ = #3  },
            xlabel={mixing parameter $\mu$},
            ylabel={quality $Q$},
            xmin=0, xmax=1,
            ymin=0, ymax=1,
            legend pos=south west,
            ymajorgrids=true,
            grid style=dashed,
        ]

        \addplot coordinates{#4};
        \addplot coordinates{#5};
        \addplot coordinates{#6};

        \legend{
            $\langle k \rangle = 20$, 
            $\langle k \rangle = 30$, 
            $\langle k \rangle = 40$, 
        }
     
    \end{axis}
    \end{tikzpicture}
}


\begin{figure}
    \centering
    \includegraphics[scale=0.20, angle=180]{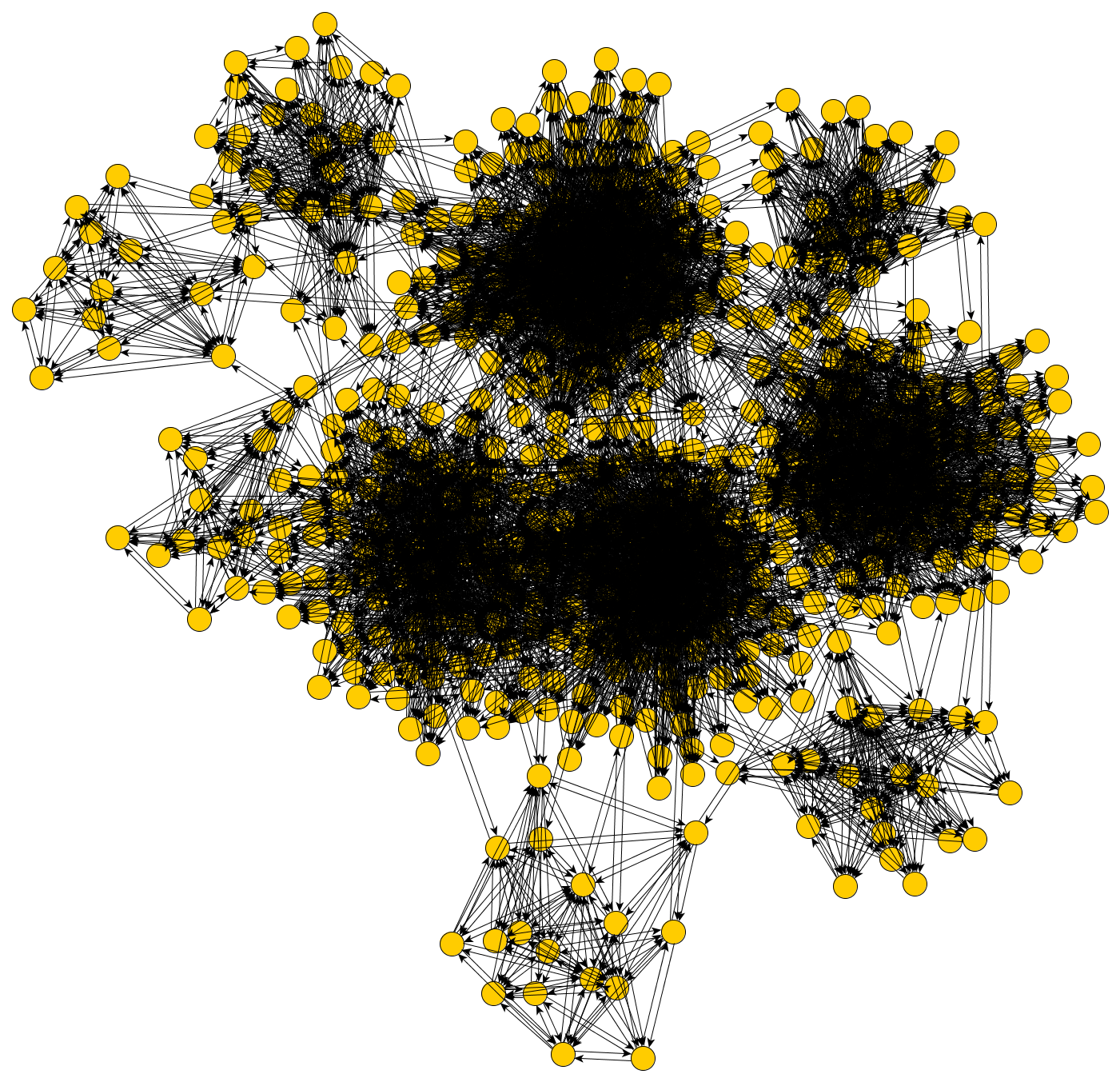}
    \caption[An example graph generated by the LFR benchmark.]
        {A graph generated by the LFR benchmark with $N = 500$, exponents $\gamma = 2$ and $\beta = 2$, average node degree $\langle k \rangle = 20$, and mixing parameter $\mu = 0.05$.
         It consists of 10 communities.} 
    \label{fig:lfrGraph}
\end{figure}

\begin{figure}
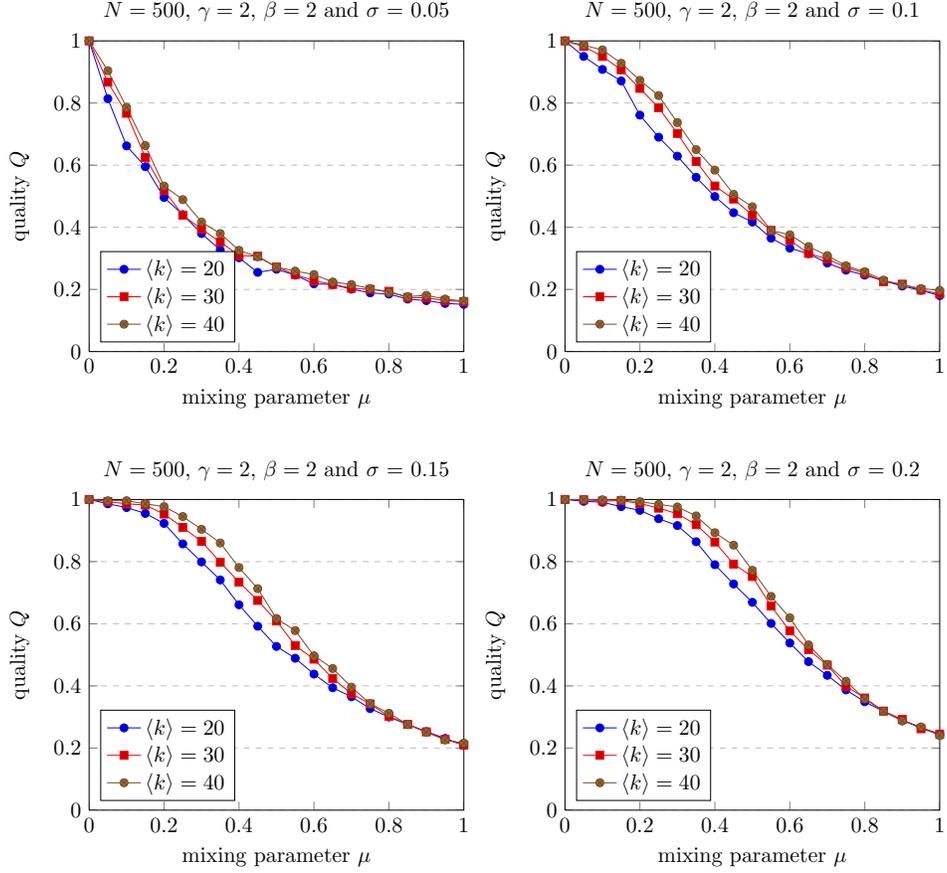

    \centering
    \begin{subfigure}{0.5\textwidth}
        \awesomeplot{2}{2}{0.05}
        {
            (0.00, 1.000)
            (0.05, 0.814)
            (0.10, 0.662)
            (0.15, 0.595)
            (0.20, 0.496)
            (0.25, 0.441)
            (0.30, 0.380)
            (0.35, 0.327)
            (0.40, 0.302)
            (0.45, 0.255)
            (0.50, 0.265)
            (0.55, 0.246)
            (0.60, 0.218)
            (0.65, 0.216)
            (0.70, 0.201)
            (0.75, 0.189)
            (0.80, 0.185)
            (0.85, 0.169)
            (0.90, 0.164)
            (0.95, 0.155)
            (1.00, 0.152)
        }
        {
            (0.00, 1.000)
            (0.05, 0.867)
            (0.10, 0.767)
            (0.15, 0.625)
            (0.20, 0.519)
            (0.25, 0.439)
            (0.30, 0.394)
            (0.35, 0.353)
            (0.40, 0.309)
            (0.45, 0.307)
            (0.50, 0.272)
            (0.55, 0.248)
            (0.60, 0.228)
            (0.65, 0.216)
            (0.70, 0.204)
            (0.75, 0.200)
            (0.80, 0.194)
            (0.85, 0.174)
            (0.90, 0.173)
            (0.95, 0.164)
            (1.00, 0.161)
        }
        {
            (0.00, 1.000)
            (0.05, 0.904)
            (0.10, 0.786)
            (0.15, 0.663)
            (0.20, 0.533)
            (0.25, 0.489)
            (0.30, 0.417)
            (0.35, 0.380)
            (0.40, 0.326)
            (0.45, 0.305)
            (0.50, 0.272)
            (0.55, 0.259)
            (0.60, 0.248)
            (0.65, 0.224)
            (0.70, 0.216)
            (0.75, 0.203)
            (0.80, 0.190)
            (0.85, 0.177)
            (0.90, 0.181)
            (0.95, 0.169)
            (1.00, 0.163)
        }
    \end{subfigure}%
    \begin{subfigure}{0.5\textwidth}
        \awesomeplot{2}{2}{0.1}
        {
            (0.00, 1.000)
            (0.05, 0.950)
            (0.10, 0.908)
            (0.15, 0.871)
            (0.20, 0.761)
            (0.25, 0.690)
            (0.30, 0.629)
            (0.35, 0.561)
            (0.40, 0.499)
            (0.45, 0.447)
            (0.50, 0.417)
            (0.55, 0.365)
            (0.60, 0.333)
            (0.65, 0.315)
            (0.70, 0.285)
            (0.75, 0.262)
            (0.80, 0.246)
            (0.85, 0.228)
            (0.90, 0.211)
            (0.95, 0.197)
            (1.00, 0.180)
        }
        {
            (0.00, 1.000)
            (0.05, 0.982)
            (0.10, 0.950)
            (0.15, 0.907)
            (0.20, 0.847)
            (0.25, 0.785)
            (0.30, 0.702)
            (0.35, 0.612)
            (0.40, 0.533)
            (0.45, 0.491)
            (0.50, 0.439)
            (0.55, 0.390)
            (0.60, 0.358)
            (0.65, 0.316)
            (0.70, 0.298)
            (0.75, 0.271)
            (0.80, 0.252)
            (0.85, 0.225)
            (0.90, 0.217)
            (0.95, 0.198)
            (1.00, 0.186)
        }
        {
            (0.00, 1.000)
            (0.05, 0.986)
            (0.10, 0.971)
            (0.15, 0.928)
            (0.20, 0.873)
            (0.25, 0.824)
            (0.30, 0.737)
            (0.35, 0.650)
            (0.40, 0.584)
            (0.45, 0.506)
            (0.50, 0.466)
            (0.55, 0.390)
            (0.60, 0.376)
            (0.65, 0.338)
            (0.70, 0.309)
            (0.75, 0.276)
            (0.80, 0.257)
            (0.85, 0.230)
            (0.90, 0.215)
            (0.95, 0.203)
            (1.00, 0.197)
        }
    \end{subfigure}
    \begin{subfigure}{0.5\textwidth}
        \awesomeplot{2}{2}{0.15}
        {
            (0.00, 1.000)
            (0.05, 0.986)
            (0.10, 0.974)
            (0.15, 0.955)
            (0.20, 0.923)
            (0.25, 0.857)
            (0.30, 0.799)
            (0.35, 0.741)
            (0.40, 0.661)
            (0.45, 0.592)
            (0.50, 0.527)
            (0.55, 0.489)
            (0.60, 0.438)
            (0.65, 0.394)
            (0.70, 0.365)
            (0.75, 0.327)
            (0.80, 0.300)
            (0.85, 0.276)
            (0.90, 0.253)
            (0.95, 0.231)
            (1.00, 0.209)
        }
        {
            (0.00, 1.000)
            (0.05, 0.993)
            (0.10, 0.986)
            (0.15, 0.982)
            (0.20, 0.953)
            (0.25, 0.910)
            (0.30, 0.865)
            (0.35, 0.798)
            (0.40, 0.734)
            (0.45, 0.675)
            (0.50, 0.609)
            (0.55, 0.530)
            (0.60, 0.486)
            (0.65, 0.424)
            (0.70, 0.375)
            (0.75, 0.342)
            (0.80, 0.302)
            (0.85, 0.276)
            (0.90, 0.252)
            (0.95, 0.228)
            (1.00, 0.210)
        }
        {
            (0.00, 1.000)
            (0.05, 0.996)
            (0.10, 0.996)
            (0.15, 0.986)
            (0.20, 0.977)
            (0.25, 0.945)
            (0.30, 0.904)
            (0.35, 0.860)
            (0.40, 0.781)
            (0.45, 0.713)
            (0.50, 0.617)
            (0.55, 0.578)
            (0.60, 0.497)
            (0.65, 0.456)
            (0.70, 0.396)
            (0.75, 0.343)
            (0.80, 0.312)
            (0.85, 0.276)
            (0.90, 0.251)
            (0.95, 0.226)
            (1.00, 0.216)
        }
    \end{subfigure}%
    \begin{subfigure}{0.5\textwidth}
        \awesomeplot{2}{2}{0.2}
        {
            (0.00, 1.000)
            (0.05, 0.994)
            (0.10, 0.991)
            (0.15, 0.977)
            (0.20, 0.965)
            (0.25, 0.938)
            (0.30, 0.916)
            (0.35, 0.864)
            (0.40, 0.790)
            (0.45, 0.728)
            (0.50, 0.669)
            (0.55, 0.601)
            (0.60, 0.538)
            (0.65, 0.478)
            (0.70, 0.434)
            (0.75, 0.387)
            (0.80, 0.349)
            (0.85, 0.318)
            (0.90, 0.290)
            (0.95, 0.266)
            (1.00, 0.245)
        }
        {
            (0.00, 1.000)
            (0.05, 0.999)
            (0.10, 0.995)
            (0.15, 0.996)
            (0.20, 0.987)
            (0.25, 0.972)
            (0.30, 0.954)
            (0.35, 0.919)
            (0.40, 0.863)
            (0.45, 0.792)
            (0.50, 0.752)
            (0.55, 0.657)
            (0.60, 0.577)
            (0.65, 0.517)
            (0.70, 0.467)
            (0.75, 0.399)
            (0.80, 0.361)
            (0.85, 0.319)
            (0.90, 0.292)
            (0.95, 0.262)
            (1.00, 0.245)
        }
        {
            (0.00, 1.000)
            (0.05, 1.000)
            (0.10, 0.999)
            (0.15, 0.998)
            (0.20, 0.993)
            (0.25, 0.984)
            (0.30, 0.976)
            (0.35, 0.947)
            (0.40, 0.893)
            (0.45, 0.853)
            (0.50, 0.772)
            (0.55, 0.688)
            (0.60, 0.619)
            (0.65, 0.532)
            (0.70, 0.469)
            (0.75, 0.415)
            (0.80, 0.358)
            (0.85, 0.318)
            (0.90, 0.287)
            (0.95, 0.268)
            (1.00, 0.240)
        }
    \end{subfigure}
    \caption[Results of the benchmark]
    {Results of the benchmark. The plotted value for $Q$ is the average over 100 runs with different random seed vertices and graphs generated for the set of parameters.}
    \label{fig:resultBenchmark}
\end{figure}

\begin{figure}
    \centering
    \begin{tabular}{| c | l |}
        \hline
        $\gamma$            & Exponent of power law distribution of node degree \\ \hline
        $\beta$             & Exponent of power law distribution of community size \\ \hline
        $\langle k \rangle$ & average node degree \\ \hline
        $N$                 & Number of nodes \\ \hline
        $\sigma$            & Fraction of nodes that are seed nodes \\ \hline
        $\mu$               & Mixing parameter \\ \hline
        $Q$                 & quality measure \\ \hline
    \end{tabular}
    \caption{Summary of all parameters for the benchmark}
    \label{fig:summaryBenchmark}
\end{figure}

\begin{figure}
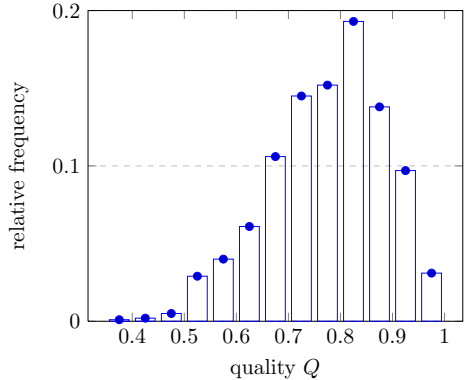

    \centering
    \begin{subfigure}{0.5\textwidth}
        \varianceplot{2}{2}{0.1}{20}{0.2}
        {
            ( 0.375,  0.001  )
            ( 0.425,  0.002  )
            ( 0.475,  0.005  )
            ( 0.525,  0.029  )
            ( 0.575,  0.04   )
            ( 0.625,  0.061  )
            ( 0.675,  0.106  )
            ( 0.725,  0.145  )
            ( 0.775,  0.152  )
            ( 0.825,  0.193  )
            ( 0.875,  0.138  )	
            ( 0.925,  0.097  )
            ( 0.975,  0.031  )
        }

    \end{subfigure}
    \caption[Histogram over the distribution of $Q$]
    {Histogram over the distribution of $Q$. 1000 data points created by 1000 runs on the same graph. The seed nodes were chosen randomly for each iteration.}
    \label{fig:histogram}

    \label{fig:bla}
\end{figure}


\subsection{Evaluation}

The benchmark setup we used to evaluate our model heavily relies on the methods proposed and used by Lancichinetti et al.
In their paper \textit{Community detection algorithms: a comparative analysis}~\cite{lan2009} 
they employ the LFR benchmark to evaluate modern methods for overlapping and non-overlapping community detection.
They use the less intuitive \textit{Normalized Mutual Information} as quality measure and slightly different parameters for graph generation,
so one has to be careful when comparing our results with theirs, however, their paper may help interpreting our results.
Further studies should strive to recreate their test setup to compare our method to state-of-the-art methods.

\paragraph{}
The results of the benchmark in figure~\ref{fig:summaryBenchmark} show that the performance of our model greatly depends on the choice of seed nodes.
The more seed nodes are available, the better is the quality of the output.
Given a network with mixing parameter $\mu = 0.3$ the algorithm fails when given only 5\% seed nodes ($Q \approx 0.4$) 
but works fine when given 20\% seed nodes ($Q \approx 0.95$).

There is also a correlation between the mixing parameter of the graph and the fraction of seed nodes needed to archive an accurate detection.
The more diffuse the communities are, the more seed nodes are needed for good detection.
To archive a quality of $Q = 0.95$ in a graph with mixing parameter $\mu = 0.1$ a fraction of 10\% seed nodes are needed.
If $\mu = 0.3$ it takes 20\% seed nodes to archive the same quality.
However, if the mixing parameter of the graph becomes too big ($\mu > 0.5$) the algorithm fails even for large fractions of seed nodes.
Also high average vertex degrees slightly improve the quality of the detection.

Furthermore, figure~\ref{fig:histogram} shows that for the same graph and different sets of seed nodes of the same size the quality of the detection varies a lot.
So the choice of good seed nodes is critical for the performance of the algorithm.
It could be an interesting subject for further studies to identify criteria for the choice of good seed nodes.

\newpage
\section{Conclusion}

The goal of this work was the development and evaluation of an algorithm for overlapping, fuzzy community detection based on random walks.
The algorithm extracts specific communities from a network based on a set of seed nodes. 
It runs in time near-linear in the number of communities times the number of edges in the network.
We used the LFR benchmark for evaluation and found that, given a good set of seed nodes, the algorithm is able to correctly reconstruct the communities of a network.

\paragraph{}
The number of communities contributes linearly to the algorithm's running time.
This means that in a scenario where the number of communities is proportional to the number of nodes, the algorithm's running time would be squared in the number of nodes,
thus may not be efficient.
It may be better suited for scenarios with only a low number of communities, such as the
detection of political leanings of users of a social network.

\paragraph{}
Further work should focus on a rigorous analysis of the proposed method based on the LFR benchmark. 
One should aim to evaluate the performance under different settings, as well
as to compare the algorithm to other methods for community detection.
It could also be of interest to test the usefulness of our model in real world situations or
to analyze different strategies for choosing good seed nodes.

\newpage
\bibliography{thesisCommunityDetection}{}
\bibliographystyle{plain}

\end{document}